\documentclass[twocolumn,10pt]{revtex4-1}
\usepackage[utf8]{inputenc}  
\usepackage[T1]{fontenc}     
\usepackage[british]{babel}  
\usepackage[scaled=0.86]{berasans}  
\usepackage[colorlinks=true,allcolors=blue]{hyperref}  
\usepackage{graphicx} 
\usepackage[babel]{microtype}  
\usepackage{amsmath,amssymb,amsthm,bm,amsfonts,mathrsfs,bbm} 
\usepackage{setspace}
\usepackage{braket}
\usepackage{csquotes}

\usepackage{xspace}  
\usepackage{pgfplots}
\usepackage{verbatim}

\def\proj{{\hat{\cal P}}}

\newcommand\id{\leavevmode\hbox{\small1\kern-3.3pt\normalsize1}}

\newtheorem{theorem}{Theorem}

\begin{document}

\title{Several Trade off Features of Quantum Steering in Distributed Scenario}

\author{Arup Roy}
\email{arup145.roy@gmail.com}
\affiliation{Centre for Astroparticle Physics and Space Science(CAPSS), Bose Institute, Block EN, Sector V, Salt Lake, Kolkata-700091, India}

\author{Some Sankar Bhattacharya}
\email{somesankar@gmail.com}
\affiliation{Department of Computer Science, The University of Hong Kong, Pokfulam Road, Hong Kong}

\author{Amit Mukherjee}
\email{amitisiphys@gmail.com}
\affiliation{Optics and Quantum Information Group,The Institute of Mathematical Sciences,HBNI,C.I.T Campus,Taramani, Chennai-600113,India}

\author{Nirman Ganguly}
\email{nirmanganguly@gmail.com}
\affiliation{Department of Mathematics,Heritage Institute of Technology, Kolkata- 700107, India}

\author{Biswajit Paul}
\email{biswajitpaul4@gmail.com}
\affiliation{Department of Mathematics, South Malda College, Malda, West Bengal, India}

\author{Kaushiki Mukherjee}
\email{kaushiki_mukherjee@rediffmail.com}
\affiliation{Department of Mathematics, Government Girls' General Degree College, Ekbalpore, Kolkata-700023, India}

\begin{abstract}
	In the present work we address the question of how bipartite steering violation takes place among multi-partite systems (where each sub-system have Hilbert space dimension restricted to two) based on the maximal violations of the bipartite steering inequality of the reduced pairwise qubit systems. We have derived  a trade-off relation which is satisfied by those pair wise bipartite maximal steering violations, which physically can be understood as providing restrictions on the distribution of steering  among sub systems. For a three-qubit system, it is impossible that all pairs of qubits violate the steering inequality, and once a pair of qubits violates the steering inequality maximally, the other two pairs of qubits must both obey the steering inequality. We also present a complementarity relation between genuine entanglement present in a tripartite state and maximum bipartite steering violation by its reduced states.
\end{abstract}
\maketitle
\section{Introduction}
In 1935 Einstein, Podolsky and Rosen presented an argument regarding the incompleteness of Quantum Mechanics \cite{epr} which motivated Schr\"{o}dinger to conceive the celebrated concept of `steering'\cite{schr} in his famous paper named 'Spooky action at a distance'. Recently, Wiseman \emph{et al.} have developed this phenomena in the form of an operational task\cite{wisemanprl07,wisemanpra07}. They have argued that steering refers to a scenario where one party usually called Alice, wishes to convince the other party Bob that she can steer or construct the conditional states on Bob's side by making
measurements on her part. The aforementioned interpretation has recently stimulated interest in foundational research \cite{jevtic2014,rudolph2014,jevtic12014,brunneroneway2014,wittman2014,wiseman2014}.

\par
Apart form the foundational interest, the study of steering also finds applications in semi device independent scenario where only one party trusts his/her quantum device but the other party's device is untrusted. As a concrete example it has been shown that steering allows for secure quantum key distribution when one of the parties' device cannot be trusted. One big advantage in this direction is that such scenarios are experimentally less demanding than fully device-independent protocols (where both of the parties distrust their devices) and, at the same time, require less assumptions than standard quantum cryptographic scenarios.
\par
In 1964, John Bell sought a way to demonstrate that certain correlations appearing in quantum mechanics are incompatible with the notions of locality and reality aka local-realism, through an inequality involving measurement statistics only. Violation of such an inequality implies the usefulness of correlations for EPR argument. In 1969, Clauser-Horne-Shimony-Holt (CHSH) proposed a set of simple Bell inequalities which are easy to realize experimentally. In the same spirit of Bell's inequality in nonlocality, several steering inequalities (SIs) have been proposed \cite{reid,ou,caval,wal,wal1}, so that a violation of any such SI can render a correlation to be steerable. But an unavoidable hindrance to formalize such SIs follows from the fact that steering scenario is device-independent only on one-side. Recently Cavalcanti et al have proposed a CHSH-like inequality \cite{caval15} for quantum steering. They have derived an EPR-SI that is necessary and sufficient for a set of correlations in
the simplest two-party scenario involving two measurement settings per site and two outcomes per measurement, with mutually unbiased measurements by trusted party.

\par
One significant departure of quantum entanglement from classical correlations is monogamy. Namely, if a quantum system is entangled with another system then this imposes a restriction on its entanglement with the remaining systems. Several distributions of quantum entanglement have been noticed in multipartite scenarios\cite{Bennett96, Terhal04, Coffman00, Kim12}. Monogamy also constitutes a significant characteristic in the security in quantum key distribution\cite{Gisin02}.  Several studies have been carried out in this important area\cite{Toner,bell-monogamy}.Mainly, the trade off relations between Bell violations of pairwise qubit systems have been probed in the works.
\par
In the present work,by using a steering inequality based on three measurement settings, we lay down a distribution of the non-classical feature within multiqubit systems. In this context we have derived an analytical trade-off relation obeyed by the steering test of pairwise qubits in a three-qubit system. Furthermore, we prove the existence of a complementary relation between a genuine tripartite entanglement measure and maximal bipartite steering inequality violation followed by the status of trade off relation when the reduced bipartite states are subjected to global unitary operations or local filtering operations.
\par
 The rest of the article is divided as follows: in Sec.\ref{secsteer} we have gone through the basic notion of steering. In Sec.\ref{sectang} we have discussed about some genuine entanglement measure whereas local filtering operations have been briefly discussed in Sec.\ref{secfil}. In Sec.\ref{secresults} we have derived the trade off relation between bipartite steering violation of three qubit systems. Then we have analytically established the complementary relation between genuine entanglement measure vs bipartite steering inequality violation in Sec.\ref{sectrade1}. In Sec.\ref{sectrade2} and Sec.\ref{sectrade4} we have examined the status of this bipartite steering and Bell inequality violation in the scenario where each bipartite system is allowed to perform global unitary and local filtering operations.

\section{Brief review of Steering}\label{secsteer}
Let us first briefly review the steering scenario as introduced by Wiseman \emph{et al.} \cite{wisemanprl07,wisemanpra07}. Given a pair of systems at Alice and Bob, denote $\mathcal{D}_{\alpha}$ and $\mathcal{D}_{\beta}$ the sets of observables in the Hilbert space of Alice's and Bob's systems, respectively. An element of $\mathcal{D}_{\alpha}$ is denoted by $A$, with a set of outcomes labeled by $a\in\mathcal{L}(A)$, and similarly for Bob. The joint state $\rho_{AB}$ of the system is steerable by Alice iff it is not the case that for all $a\in\mathcal{L}(A)$, $b\in\mathcal{L}(B)$, $A\in\mathcal{D}_{\alpha}$, $B\in\mathcal{D}_{\beta}$, the joint probability distributions can be written in the form
\begin{equation}\label{steer}
P(a,b|A,B;\rho_{AB})=\sum_{\lambda}\wp(\lambda)\wp(a|A,\lambda)P(b|B;\rho_{\lambda}),
\end{equation}
where $\wp(a|A,\lambda)$ denotes an arbitrary probability distribution and $P(b|B;\rho_{\lambda})$ denotes the quantum probability of outcome $b$ given measurement $B$ on state $\rho_{\lambda}$. In other words the state $\rho_{AB}$ will be called steerable if it does not satisfy a LHV-LHS model. Note that, if for a given measurement strategy the correlation has a LHV-LHS model, this does not imply that the underlying state is not steerable, since there could be another strategy that does not. In \cite{caval} authors have developed a series of steering inequalities to check whether a bipartite state is steerable when both the parties are allowed to perform $\mathit{n}$  measurements on his or her part.
\begin{equation}\label{steering n}
F_n(\rho,\mu) =\frac{1}{\sqrt{n}} \Big|\sum_{i=1}^n\langle A_i\otimes B_i\rangle\Big| \leqslant 1,
\end{equation}

The inequalities for
$\mathit{n}=2,3$ are of the form:
\begin{equation}\label{steering 2}
F_2(\rho,\mu) =\frac{1}{\sqrt{2}} \Big|\sum_{i=1}^2\langle A_i\otimes B_i\rangle\Big| \leqslant 1,
\end{equation}

\begin{equation}\label{steering 3}
F_3(\rho,\mu) =\frac{1}{\sqrt{3}} \Big|\sum_{i=1}^3\langle A_i\otimes B_i\rangle\Big| \leqslant 1,
\end{equation}
where $A_i = \hat{u}_i\cdot\vec{s}$,  $B_i = \hat{v}_i\cdot\vec{s}$, $\vec{s}=(s_1,s_2,s_3)$ is a vector composed of the Pauli matrices, $\hat{u}_i\in\mathbb{R}^3$ are unit vectors, $\hat{v}_i\in\mathbb{R}^3$ are orthonormal vectors, $\mu=\{\hat{u}_1,\cdots,\hat{u}_n,\hat{v}_1,\cdots,\hat{v}_n\}$ is the set of measurement directions, $\langle A_i\otimes B_i\rangle=\text{Tr}(\rho A_i\otimes B_i)$, and $\rho\in\mathcal{H}_A\otimes\mathcal{H}_B$ is some bipartite quantum state.

\section{Tangle: A measure of genuine entanglement}\label{sectang}
One of the significant measure of genuine entanglement in three qubit systems is the tangle. The notion of tangle relies on the monogamy of quantum correlations\cite{Coffman00,monogamybunch}.
The tangle \cite{Coffman00} is defined as the quantum monogamy score \cite{mdcc} corresponding to the square of the bipartite entanglement measure, called the concurrence \cite{wootters-concurrence}.
Concurrence of a two qubit system is given as
$C(\rho_{AB}) = \max\{0,\lambda_1-\lambda_2-\lambda_3-\lambda_4\}$.
Here $\lambda_1,\dots,\lambda_4$ are the square roots of the eigenvalues of $\rho_{AB}\tilde{\rho}_{AB}$ in
decreasing order, $\tilde{\rho}_{AB}= (\sigma_y \otimes \sigma_y)\rho_{AB}^{*}(\sigma_y \otimes \sigma_y)$.
The complex conjugation $\rho_{AB}^{*}$ is taken in the computational basis, and $\sigma_y$ is the Pauli spin matrix.

Therefore, the tangle of a three-qubit state $\rho_{ABC}$ is given by \cite{Coffman00}
\begin{equation}
	\tau(\rho_{ABC}) = C^2_{A:BC} - C^2_{AB} - C^2_{AC}.
	\label{eq:tangle}
\end{equation}
\section{Brief review of Local Filtering Operations}\label{secfil}
 Other than Bell-CHSH nonlocality tests, nonlocality of a bipartite quantum system can be exploited by subjecting it to a sequence of measurements\cite{pop}. Local filtering operations form a particular type of sequential measurements and nonlocality of this form is called hidden nonlocality. A necessary and sufficient criterion for detecting hidden nonlocality of $\rho_{AB}$ with respect to Bell-CHSH inequality was derived in \cite{hidebel}. The criterion is discussed here.\\
Let $M$ denotes a $4\times 4$ real matrix with elements $M_{ij}$ where $M_{ij}$ are given by:\\ $M_{ij}$$=$$\textmd{Tr}(\rho_{AB}\sigma_{i}\otimes\sigma_{j}),\,i,j$$=$$0,1,2,3$ with $\sigma_0$ denoting $\mathbb{I}_2$.\\ Let $Q$$=$$\textmd{diag}(1,-,1,-1,-1)$ and $G_{\rho}^{AB}$$=$$QMQM^T.$ \\
$\rho_{AB}$ reveals hidden nonlocality with respect to Bell-CHSH inequality \textit{iff}:
\begin{equation}\label{filtercon}
    \Lambda_1(G_{\rho}^{AB})+\Lambda_2(G_{\rho}^{AB})>\Lambda_0(G_{\rho}^{AB}),
\end{equation}
where $\Lambda_i(G_{\rho}^{AB})(i=0,1,2,3)$ are the eigen values of $G_{\rho}^{AB}$ in descending order.
After discussing these quantitative preliminaries, we now organize our findings. We start by exploring restrictions(if any) over the share-ability of a non-classical correlation among the reduced bipartite states obtained from a tripartite quantum state.

\section{Trade off relations in steering inequalities}\label{secresults}
 Here we look at the trade-off relations regarding the steering inequality(\ref{steering n}) violation by the reduced states of any quantum state shared between three parties in 2-settings per site and 3-settings per site scenario:
\begin{theorem}: For any tripartite state $\rho_{ABC}\in \mathcal{H}_A\small{\bigotimes}\mathcal{H}_B\small{\bigotimes}\mathcal{H}_C$, shared between Alice, Bob and Charlie, the maximal violation of $2$ settings steering inequality(Eq.(\ref{steering 2})) by the pairwise bipartite reduced states satisfy the following trade-off relation:
\begin{equation}\label{s1}
\langle F_2^S\rangle_{AB}^2+\langle F_2^S\rangle_{AC}^2+\langle F_2^S\rangle_{BC}^2\leq 3.
\end{equation}
\end{theorem}
\begin{proof}
For any bipartite state $\rho_{AB}$(say), square of maximal violation of $2$ setting steering inequality(Eq.(\ref{steering 2})) is given in \cite{costa}:
\begin{equation}\label{s2}
\langle F_2^S\rangle_{AB}^2=\sum_{i=1}^3t_i^2-\min\{t_1^2,t_2^2,t_3^2\}
\end{equation}
where $t_i$ are singular values of the correlation matrix $T$$=$$(t_{ij}^{AB})_{3\times 3}$ of the state $\rho_{AB}.$
Clearly $\langle F_2^S\rangle_{AB}^2=\frac{B^2_{\textmd{Max}}(\rho_{AB})}{4}$, where $B_{\textmd{Max}}(\rho_{AB})$ is maximal Bell-CHSH violation\cite{horodecki}. Now $B_{\textmd{Max}}^2(\rho_{AB})$ is given in \cite{fei}:
$$B_{\textmd{Max}}^2(\rho_{AB})=4(\sum_{i,j=1}^3t_{ij}^{AB}(t_{ij}^{AB})^{*}-\min\{t_1^2,t_2^2,t_3^2\})$$
\begin{equation}\label{s3}
\leq 4\sum_{i,j=1}^3t_{ij}^{AB}(t_{ij}^{AB})^{*}.
\end{equation}
Using Eq.(\ref{s3}) for each of the three reduced states $\rho_{AB}$, $\rho_{BC}$ and $\rho_{AC}$, in \cite{fei} author has derived the monogamy relation for maximal Bell-CHSH violation:
\begin{equation}\label{F}
B_{\textmd{Max}}^2(\rho_{AB})+B_{\textmd{Max}}^2(\rho_{BC})+B_{\textmd{Max}}^2(\rho_{AC})\leq 12.
\end{equation}
This in turn implies the monogamy relation for $2$ setting steering inequality(Eq.(\ref{s1})).
\end{proof}
A change in the number of measurement settings create a drastic change in Bell CHSH scenario. States which are Bell nonlocal under 3 measurement settings do not violate a Bell inequality with two measurement settings. Hence, we probe below the trade off relations in the steering scenario beyond two measurement settings. 
Now if one considers 3 settings per site scenario, there exists a class of local realistic inequalities\cite{CollinsGisin04} inequivalent to the Bell-CHSH, which do not satisfy any monogamy relation. A pertinent question at this point one can ask is whether 3 setting steering inequalities of the form Eq.(\ref{steering 3}) abide by any monogamy relation or not. Interestingly we find that if the $3$ setting bipartite steering inequality is considered, a more stringent monogamy relation is obtained: \\
\begin{theorem}: For tripartite state $\rho_{ABC}$, shared between Alice, Bob and Charlie, the maximal violation of $3$ settings steering inequality(Eq.(\ref{steering 3})) by the pairwise bipartite reduced states satisfies the following trade-off relation:
\begin{equation}\label{s4}
\langle F_3^S\rangle_{AB}^2+\langle F_3^S\rangle_{AC}^2+\langle F_3^S\rangle_{BC}^2\leq 3.
\end{equation}
\end{theorem}
For pure state this trade off relation is strict. But for mixed state it is not srtict.
\begin{proof}

Case 1:

For any bipartite reduced state $\rho_{AB}$ of a pure state $|\psi\rangle_{ABC}$,
$$ \langle F_3^S\rangle_{AB}^2=\sum_{i=1}^3t_i^2$$
$$=\langle F_2^S\rangle_{AB}^2+\min\{t_1^2,t_2^2,t_3^2\}$$
$$=\frac{B_{\textmd{Max}}^2(\rho_{AB})}{4}+\min\{t_1^2,t_2^2,t_3^2\}$$
$$=(\sum_{i,j=1}^3t_{ij}^{AB}(t_{ij}^{AB})^{*}-\min\{t_1^2,t_2^2,t_3^2\})+\min\{t_1^2,t_2^2,t_3^2\}$$
\begin{equation}\label{s5}
=\sum_{i,j=1}^3t_{ij}^{AB}(t_{ij}^{AB})^{*}
\end{equation}
Hence, from \cite{fei}
$$  \langle F_3^S\rangle_{AB}^2+\langle F_3^S\rangle_{AC}^2+\langle F_3^S\rangle_{BC}^2$$
$$=\sum_{i,j=1}^3t_{ij}^{AB}(t_{ij}^{AB})^{*}+\sum_{i,j=1}^3t_{ij}^{BC}(t_{ij}^{BC})^{*}+\sum_{i,j=1}^3t_{ij}^{AC}(t_{ij}^{AC})^{*}.$$
\begin{equation}\label{s6}
=3
\end{equation}
Case 2:
For any bipartite reduced state $\rho_{AB}$ of a mixed state  $\rho_{ABC}$,
$$ \langle F_3^S\rangle_{AB}^2=\sum_{i=1}^3t_i^2$$
$$=\langle F_2^S\rangle_{AB}^2+\min\{t_1^2,t_2^2,t_3^2\}$$
$$=\frac{B_{\max}^2(\rho_{AB})}{4}+\min\{t_1^2,t_2^2,t_3^2\}$$
$$=(\sum_{i,j=1}^3t_{ij}^{AB}(t_{ij}^{AB})^{*}-\min\{t_1^2,t_2^2,t_3^2\})+\min\{t_1^2,t_2^2,t_3^2\}$$
This gives Eq.(\ref{s5}).\\
Hence, from \cite{fei}
$$  \langle F_3^S\rangle_{AB}^2+\langle F_3^S\rangle_{AC}^2+\langle F_3^S\rangle_{BC}^2$$
$$=\sum_{i,j=1}^3t_{ij}^{AB}(t_{ij}^{AB})^{*}+\sum_{i,j=1}^3t_{ij}^{BC}(t_{ij}^{BC})^{*}+\sum_{i,j=1}^3t_{ij}^{AC}(t_{ij}^{AC})^{*}.$$
\begin{equation}\label{s6}
\leq 3
\end{equation}
\end{proof}

Having derived trade-off relation restricting distribution of steering nonlocality among reduced states, it will now be interesting to look for a relation between steering nonlocality present in reduced states and the genuine entanglement content of the original tripartite state.
\section{Trade-off between Genuine entanglement and maximum steering inequality violation}\label{sectrade1}
There are clear distinctions between W and GHZ states in terms of retaining entanglement under particle loss. While W states retain entanglement , GHZ states fail to do so. In the present section we have probed whether there is a trade-off relation between genuine entanglement and steering. 

We establish a relation between the bipartite three measurements steering inequality violation as in eq. \ref{steering 3} and the genuine tripartite correlation for an important class of three qubit pure states. In particular we show that there exists a complementary relation between the genuine tripartite quantum correlation measures and the bipartite steering inequality violation of three qubit pure states akin to a similar work in Bell violation earlier\cite{mishra}. We consider a single parameter family of genuinely tripartite entangled three qubit pure states that gives the maximum bipartite Bell inequality violation for a fixed amount of tripartite correlation. This family of the genuinely entangled three qubit pure states looks like
\begin{equation}
	\label{mdcc-class-state}
	\ket{\psi}_m = \frac{\ket{000} + m(\ket{010}+\ket{101})+ \ket{111}}{\sqrt{2+2 m^2}},
\end{equation}
where $m\in[0,1]$. These states belong to the $GHZ$ class  when $m\in[0,1)$. For $m=1$, the state belongs to the $W$ class as the tangle is zero at this point. We denote this class of states as the maximally Bell inequality violating (MBV) class of states. This class of states has been recognized as the maximally dense-coding capable (MDCC) \cite{mdcc}, for having maximal dense coding capabilities with fixed amount of genuine tripartite quantum correlations. Let three parties A, B and C share the following MDCC state:
\begin{eqnarray}
\ket{\psi_m}&=&\dfrac{1}{\sqrt{2(1+m^2)}}[\ket{000}\nonumber\\
&+&m(\ket{010}+\ket{101})+\ket{111}]
\end{eqnarray}
It is straight forward to calculate the 3 setting steering inequality violations by the three reduced states corresponding to the parties AB, BC and AC:
\begin{eqnarray}
	\langle F^2_3\rangle_{AB}=\langle F^2_3\rangle_{BC}&=&(\dfrac{1-m^2}{1+m^2})^2\\
	\langle F^2_3\rangle_{AC}&=&1+\dfrac{8m^2}{(1+m^2)^2}
\end{eqnarray}\cite{mishra}.
We denote the maximum violation of the 3 setting steering inequality among the three reduced state by $\langle F^2_3\rangle_{(\proj{\psi_m})}$, which is given by:
\begin{eqnarray}\label{maxsteer}
	\langle F^2_3\rangle_{(\proj{\psi_m})}&=&\langle F^2_3\rangle_{AC}=1+\dfrac{8m^2}{(1+m^2)^2}
\end{eqnarray}
The measure of genuine tripartite entanglement, tangle for the state $\ket{\psi}_m$:
\begin{eqnarray}\label{tang}
	\tau(\proj{\psi_m})&=&1-\dfrac{4m^2}{(1+m^2)^2}
\end{eqnarray}
Using Eq.(\ref{maxsteer}) and Eq.(\ref{tang}) one obtains the following trade-off relation between the maximum 3 setting steering violation by reduced states and the genuine entanglement content of the tripartite pure state $\ket{\psi}_m$:
\begin{eqnarray}	
	2\tau(\proj{\psi_m})&+&\langle F^2_3\rangle_{(\proj{\psi_m})}=3
\end{eqnarray}
On the basis of thorough numerical investigations, we conjecture that for any pure tripartite state, the following trade-off relation holds good:
 \begin{eqnarray}\label{trade off}
   2\tau(\proj{\psi_m})&+&\langle F^2_3\rangle_{(\proj{\psi_m})}\leq 3
 \end{eqnarray}
Following the technique used in \cite{mishra}, it can be shown that the trade-off relation (Eq.(\ref{trade off})) holds good for any mixed tripartite state also.

\section{Trade off relation in the context of Partial Collusion of parties}\label{sectrade2}
In this section we consider a scenario where a subset of parties can come together and perform a joint operation on their particles. In this paradigm one can define a stronger notion of monogamy relations. This is to say, whether a trade-off relation exists among the violations of Bell-CHSH by the reduced states, when any two parties (in the simplest case) are allowed to perform a joint operation on their particles. In this context we present the following result for tripartite pure qubit states:
\begin{theorem}
	A Bell-CHSH trade-off relation holds for a pure three qubit state if and only if at most one of all three reduced single particle states is not maximally mixed.
\end{theorem}
\begin{proof}
In \cite{ganguly} authors demonstrated that the bipartite reduced state $\sigma_{ij}$ of any pure three qubit state $\ket{\psi_{ijk}}$ is absolutely Bell-CHSH local if and only if the remaining party's reduced state $\sigma_k$ is maximally mixed, where $i,j,k\in\{A,B,C\}$.

It is evident from the above result that if all the three or any two reduced single particle states are \emph{not} maximally mixed then the remaining all or two bipartite  states are not absolutely Bell-CHSH local. There exists such global unitaries whose implementation on those bipartite states will leave it as a Bell-CHSH nonlocal state. This clearly shows that the modified trade-off relation (i.e., the trade-off relation after the application of the \emph{suitable} global unitaries) does not hold here.

Consider the situation where any two (all three) of the reduced single particle states are maximally mixed, i.e., any two (all three) of the possible bipartite reduced states are absolutely Bell-CHSH local. These two (three) states, therefore, satisfy Bell-CHSH inequality after the application of global unitaries. These situations clearly satisfy the modified trade-off relation.

This provides the proof of sufficient part of the theorem. We discuss the necessary part in the following.

When the modified trade-off relation holds, \emph{either} any one of the three bipartite reduced states violates Bell-CHSH inequality \emph{or} none of those reduced bipartite states violates after the implementation of global unitaries on the reduced bipartite systems. It is therefore evident from the above mentioned result of \cite{ganguly} that to satisfy the modified trade-off relation at most one of the three reduced single particle states is not maximally mixed.
\end{proof}

\subsection*{Trade-off relation under filtering}\label{sectrade3}
Similar situation can arise when two of them are allowed to perform local filtering instead of global unitary. Then also monogamy may not hold if marginals of two of them are not maximally mixed.
\par
From the last theorem it is clear that existence of trade-off relation (after bipartite states are subjected to suitable global unitaries or local filtering operations) depends on the nature of reduced states. In this context, it will be interesting to explore whether such conclusion can be made from knowledge of the original state parameters only. We proceed with related discussions in next section.
\section{Three qubit pure states obeying trade-off relations in nonlocality and steering}\label{sectrade4}
Let $|\psi_{ABC}\rangle$ denote a pure tripartite state \cite{acin}:
\begin{equation}\label{state}
	|\psi_{ABC}\rangle=\lambda_0 |000\rangle+\lambda_1 |100\rangle+\lambda_2 |101\rangle+\lambda_3 |110\rangle+\lambda_4 |111\rangle
\end{equation}
where $\lambda_i(i=0,...,4)$ are all real and $\sum_{i=0}^4\lambda_i^2=1.$\\
\begin{theorem}
	Let $|\psi_{ABC}\rangle$ be shared between three parties Alice, Bob and Charlie. If the parties are allowed to perform local filtering operations, the usual trade-off of Bell-CHSH violation among all possible groupings of the three parties will be violated iff $\lambda_0\neq 0$ and at least one of the following criteria holds:
\begin{enumerate}
	\item  $\lambda_4<\textmd{Min}\{\lambda_2,\lambda_3\}$
	\item $(\lambda_2\lambda_3-\lambda_1\lambda_4)^2>(\lambda_0\lambda_4)^2$ and $\lambda_4<\textmd{Max}\{\lambda_2,\lambda_3\}.$
\end{enumerate}
\end{theorem}
\begin{proof}
Let $\rho_{AB},$ $\rho_{BC}$ and $\rho_{AC}$ denote the reduced bipartite mixed states. Now if Alice and Bob perform suitable local filtering operations, then $\rho_{AB}$ violates Bell-CHSH inequality iff criterion given by Eq.(\ref{filtercon}) is satisfied:
\begin{equation}\label{state1}
	4\lambda_0^2(\lambda_3^2-\lambda_4^2)>0
\end{equation}
Similarly if Alice and Charlie perform suitable local filtering operations, then $\rho_{AC}$ violates Bell-CHSH inequality iff(Eq.(\ref{filtercon})):
\begin{equation}\label{state2}
	4\lambda_0^2(\lambda_2^2-\lambda_4^2)>0
\end{equation}
Analogously, Bell-CHSH violation is observed when $\rho_{BC}$ is subjected under suitable local filtering operations iff(Eq.(\ref{filtercon})):
\begin{equation}\label{state3}
	(\lambda_2\lambda_3-\lambda_1\lambda_4)^2-(\lambda_0\lambda_4)^2>0
\end{equation}
Clearly combination of the restrictions imposed over state parameters by the necessary and sufficient criteria given by Eqs.(\ref{state1},\ref{state2},\ref{state3}) simply point out the fact that  Bell-CHSH violation is observed in at least two of three possible groupings iff at least one of the two criteria holds:
\begin{enumerate}
	\item  $\lambda_0\neq 0$ and $\lambda_4<\textmd{Min}\{\lambda_2,\lambda_3\}$
	\item $\lambda_0\neq 0$ and $(\lambda_2\lambda_3-\lambda_1\lambda_4)^2>(\lambda_0\lambda_4)^2$ and $\lambda_4<\textmd{Max}\{\lambda_2,\lambda_3\}.$
\end{enumerate}
Hence the result.
\end{proof}
Interestingly Bell-CHSH violation is observed for all possible groupings \emph{iff} $\lambda_0\neq 0,$ $\lambda_4<\textmd{Min}\{\lambda_2,\lambda_3\}$ and $(\lambda_2\lambda_3-\lambda_1\lambda_4)^2>(\lambda_0\lambda_4)^2.$\\
\begin{theorem}
	If the parties are allowed to perform global unitary operations, the usual trade-off of Bell-CHSH violation among all possible groupings of the three parties will be violated iff at least one of the following criteria holds:
\begin{enumerate}
	\item  $\lambda_0^2(\lambda_4^2+\textmd{Max}\{\lambda_3^2,\lambda_2^2\})+(\lambda_2\lambda_3-\lambda_1\lambda_4)^2<\frac{1}{4}$
	\item $\lambda_0^2(\lambda_4^2+\textmd{Min}\{\lambda_3^2,\lambda_2^2\})+(\lambda_2\lambda_3-\lambda_1\lambda_4)^2<\frac{1}{4}$ and $\lambda_0^2(\lambda_2^2+\lambda_3^2+\lambda_4^2)<\frac{1}{4}.$
\end{enumerate}
\end{theorem}
\begin{proof}
Let $\rho_{AB},$ $\rho_{BC}$ and $\rho_{AC}$ denote the reduced bipartite mixed state. Now if Alice and Bob perform suitable global unitary operations, then $\rho_{AB}$ violates Bell-CHSH inequality \emph{iff}
\begin{equation}\label{state4}
	\lambda_0^2(\lambda_2^2+\lambda_4^2)+(\lambda_2\lambda_3-\lambda_1\lambda_4)^2<\frac{1}{4}.
\end{equation}
If Alice and Charlie perform suitable global unitary operations, then $\rho_{AC}$ violates Bell-CHSH inequality \emph{iff}
\begin{equation}\label{state5}
	\lambda_0^2(\lambda_4^2+\lambda_2^2)+(\lambda_2\lambda_3-\lambda_1\lambda_4)^2<\frac{1}{4}.
\end{equation}
Analogously, Bell-CHSH violation is observed when $\rho_{BC}$ is subjected under suitable global unitary operations \emph{iff}
\begin{equation}\label{state6}
	\lambda_0^2(\lambda_4^2+\lambda_2^2+\lambda_3^2)<\frac{1}{4}.
\end{equation}
Clearly combination of the restrictions imposed over state parameters by the necessary and sufficient criteria given by Eqs.(\ref{state4},\ref{state5},\ref{state6}) implies that  Bell-CHSH violation is observed in at least two of three possible groupings \emph{iff} at least one of the following two criteria holds:
\begin{enumerate}
	\item  $\lambda_0^2(\lambda_4^2+\textmd{Max}\{\lambda_3^2,\lambda_2^2\})+(\lambda_2\lambda_3-\lambda_1\lambda_4)^2<\frac{1}{4}$
	\item $\lambda_0^2(\lambda_4^2+\textmd{Min}\{\lambda_3^2,\lambda_2^2\})+(\lambda_2\lambda_3-\lambda_1\lambda_4)^2<\frac{1}{4}$ and $\lambda_0^2(\lambda_2^2+\lambda_3^2+\lambda_4^2)<\frac{1}{4}.$
\end{enumerate}
Hence the result.
\end{proof}
Interestingly Bell-CHSH violation is observed in all possible groupings \emph{iff} $\lambda_0^2(\lambda_2^2+\lambda_3^2+\lambda_4^2)<\frac{1}{4}$ and $\lambda_0^2(\lambda_4^2+\textmd{Max}\{\lambda_3^2,\lambda_2^2\})+(\lambda_2\lambda_3-\lambda_1\lambda_4)^2<\frac{1}{4}.$
\par
\begin{theorem}
If the parties are allowed to perform global unitary operations, the usual trade-off of steering inequality violation among all possible groupings of the three parties will be violated \emph{iff} at least one of the criteria given in Theorem.2 holds.
\end{theorem}
\begin{proof}
The proof can be given in an analogous approach as that followed in the last two theorems together with the use of necessary and sufficient criterion that detects steerability of a bipartite state subjected to suitable global unitary operations \cite{ganguly}.
\end{proof}

\section*{Conclusion}\label{secconc}
Here we have analytically derived a trade off relation for three measurement settings steering violation among pairwise qubit systems for any arbitrary tripartite state. Then we have addressed the question how measure of genuine tripartite entanglement relates with three measurement settings steering violation by reduced states. We find that there exists a trade off relation between a measure of genuine tripartite entanglement and maximal violation of three settings steering inequality by bipartite reduced systems. Next we put forward a completely new question that is, what should be the status of monogamy relation if any two of the three parties are allowed to perform joint operations. In this context a necessary and sufficient criterion has been derived for three qubit pure states for satisfying Bell- CHSH monogamy relations. Questions for further study can be placed as follows- it would be interesting to see whether monogamy relations can be obtained for steering in a more than three settings scenario. Another area of natural generalization would be to consider higher dimensional systems for carrying out similar studies.
 \section*{Acknowledgment}
 The authors would like to thank Guruprasad Kar and Manik Banik for many stimulating discussions and fruitful suggestions.

\section*{Author contribution statement}
All the authors contributed equally to the manuscript.

\end{document}